\documentclass[preprint,11pt]{elsarticle}

\usepackage{amssymb}
\usepackage{amsmath}
\usepackage{hyperref}
\usepackage{tikz}
\usepackage{listings}
\usepackage{pythontex}
\usepackage{amsthm}
\usepackage{longtable}
\usepackage{float}

\newtheorem{Definition}{Definition}[section]

\newtheorem{Proposition}[Definition]{Proposition}

\newtheorem{theorem}{Theorem}[section]

\newtheorem{lemma}[theorem]{Lemma}
\newtheorem{definition}{Definition}

\usepackage{array,multirow,makecell}

\setcellgapes{1pt} \makegapedcells
\newcolumntype{R}[1]{\rangle{\raggedleft\arraybackslash }b{#1}}
\newcolumntype{L}[1]{\rangle{\raggedright\arraybackslash }b{#1}}
\newcolumntype{C}[1]{\rangle{\centering\arraybackslash }b{#1}}
\newcounter{minutes}
\divide\time by 60
\newcounter{hours}
\multiply\time by 60 \addtocounter{minutes}{-\time}

\definecolor{codeblue}{rgb}{0.25, 0.5, 0.75}
\definecolor{codegreen}{rgb}{0, 0.6, 0}
\definecolor{codegray}{rgb}{0.5, 0.5, 0.5}
\definecolor{codepurple}{rgb}{0.58, 0, 0.82}
\definecolor{backcolour}{rgb}{0.95, 0.95, 0.92}

\makeatletter
\def\ps@pprintTitle{%
 \let\@oddhead\@empty
 \let\@evenhead\@empty
 \let\@oddfoot\@empty
 \let\@evenfoot\@oddfoot
} \makeatother

\begin{document}

\begin{frontmatter}

\title{On symbol-pair distance of repeated-root constacyclic codes of length $4p^s$ over  $\mathbb{F}_{p^m}+u\mathbb{F}_{p^m}+u^2\mathbb{F}_{p^m}$}

\author[VB]{Payel Chandra}
\ead{Payelchandra12@gmail.com}

\author[VB]{Kalyan Hansda\corref{cor1}}
\ead{kalyanh4@gmail.com}

\cortext[cor1]{Corresponding author.}

\address[VB]{Department of Mathematics, Visva-Bharati, Santiniketan 731235, West Bengal, India}

\begin{abstract}

This paper completely determines the symbol-pair distance distributions of all repeated-root $\Delta$-constacyclic codes of length $4p^{s}$ over the finite commutative chain ring $R_{3}=\mathbb{F}_{p^{m}}[u]/\langle u^{3}\rangle$, where $p^{m}\equiv1 \pmod 4$. The distance characterization is explicitly classified according to the quadratic character of the shift unit $\Delta \in R_{3}^{*}$. When $\Delta$ is a non-square unit, the exact symbol-pair distances are established across all eight distinct ideal classifications of the ambient ring. Conversely, when $\Delta$ is a square unit, the distance profiles are derived by evaluating direct sum decompositions and local ring reductions. By evaluating the symbol-pair singleton bound, we prove that only the trivial ideal $\mathcal{C}=\langle1\rangle$ achieves maximum distance separability (MDS) , as structural constraints rule out any non-trivial MDS configurations. Finally, computational examples of length 20 over $\mathbb{F}_{5}+u\mathbb{F}_{5}+u^{2}\mathbb{F}_{5}$ are provided to validate the derived distance formulas.
\end{abstract}

\end{frontmatter}

\section{Introduction}
Constacyclic codes form a structurally significant class of linear codes owing to their algebraic properties and efficient error-control implementations \cite{Ling2004}. Let $R$ be a finite commutative ring with identity and $\lambda \in R^{*}$ be a unit. A linear code $C$ of length $n$ over $R$ is $\lambda$-constacyclic if it is an ideal of the quotient ring $R[x]/\langle x^n - \lambda \rangle$. When the code length $n$ is divisible by the characteristic $p$ of the residue field, these codes are classified as repeated-root constacyclic codes. Although repeated-root codes are generally asymptotically bad, they possess a rich algebraic architecture and frequently produce optimal or near-optimal parameters.
 
Following the seminal finding that major nonlinear binary codes can be realized as Gray images of linear codes over finite rings, extensive research has focused on codes over finite chain rings. Particular attention has been devoted to polynomial quotient rings of the form $R_k = \mathbb{F}_{p^m}[u]/\langle u^k \rangle$. For $k=2$, Dinh \cite{Dinh2010} completely characterized repeated-root constacyclic codes of length $p^s$ over $\mathbb{F}_{p^m} + u\mathbb{F}_{p^m}$. This foundation was subsequently extended to lengths of $2p^s$ \cite{Chen2016} and $4p^s$ \cite{Dinh4p}. For the chain ring $R_3 = \mathbb{F}_{p^m}[u]/\langle u^3 \rangle$, the classification of repeated-root constacyclic codes of lengths $p^s$ and $2p^s$ was established by Laaouine and Charkani \cite{Laaouine2021} and Sriwirach and Klin-Eam \cite{Sriwirach2021}, respectively. Recently, Laaouine and Dinh \cite{Laaouine2025} determined the algebraic structure of all $\alpha$-constacyclic codes of length $4p^s$ over $R_3$ under the field condition $p^m \equiv 1 \pmod 4$.

Parallel to these algebraic developments, Cassuto and Blaum \cite{Cassuto2011} introduced symbol-pair read channels to mitigate errors in high-density storage systems—such as holographic or high-density magnetic recording media—where the physical resolution of the reading mechanism is lower than the writing resolution. In such channels, isolating individual memory units becomes unfeasible, making the standard Hamming distance non-optimal. The symbol-pair distance $d_{sp}$ addresses this limitation by measuring correctability over overlapping, consecutive symbol pairs, thereby protecting the channel against a bounded number of pair-level corruptions rather than tracking isolated coordinates. 

The fundamental geometric advantage of this metric lies in its relationship to the classical Hamming distance $d_H$; for any distinct codewords, the distance profiles are strictly bounded by $d_H(\mathbf{x}, \mathbf{y}) + 1 \leq d_{sp}(\mathbf{x}, \mathbf{y}) \leq 2d_H(\mathbf{x}, \mathbf{y})$. Consequently, a code with a moderate Hamming distance can yield a substantially larger symbol-pair distance, providing significantly stronger error-detection capability. For instance, a code with $d_H(\mathcal{C}) = 3$ can produce $d_{sp}(\mathcal{C}) = 5$, upgrading its performance from detecting at most $2$ isolated symbol errors to detecting up to $4$ pair-errors. Furthermore, while classical MDS codes governed by the Hamming metric face severe length constraints over fixed finite alphabets, symbol-pair MDS codes can achieve greater block lengths while maintaining strict minimum distance optimality. 

This enhanced capability has motivated the computation of symbol-pair distance profiles across various code families. The symbol-pair distance distributions for repeated-root codes over finite fields have been determined for lengths $p^s$ \cite{Dinh2017}, $2p^s$ \cite{Dinh2019symbol}, and $4p^s$ \cite{Dinh2025symbol}. Over finite chain rings, symbol-pair weights have been established for lengths $2p^s$ over $\mathbb{F}_{p^m} + u\mathbb{F}_{p^m}$ \cite{Dinh2023symbol}, prime-power lengths over $\mathbb{F}_{p^m} + u\mathbb{F}_{p^m}$ \cite{Dinh2019}, and prime-power lengths over $\mathbb{F}_{p^m} + u\mathbb{F}_{p^m}+u^2\mathbb{F}_{p^m}$ \cite{Charkani2021}. 

Building on these frameworks, this paper evaluates the symbol-pair metric for repeated-root codes over $R_3$ of length $4p^s$. Under the constraint $p^m \equiv 1 \pmod 4$, we explicitly compute the symbol-pair distance of all $\Delta$-constacyclic codes of length $4p^s$ over $R_3$. Furthermore, we establish the necessary and sufficient algebraic conditions for these configurations to yield maximum distance separable (MDS) symbol-pair codes in accordance with the symbol-pair Singleton bound.

The rest of this paper is structured as follows. Section 2 introduces the fundamental algebraic properties of the chain ring $R_3$ and symbol-pair metrics. Section 3 details the core structural computation of the symbol-pair distance of $\Delta$-constacyclic codes. Section 4 classifies the parameters required to achieve MDS symbol-pair codes alongside illustrative examples, and Section 5 concludes the paper .

\section{Preliminaries}
Throughout this paper, we will serves $R$ as a finite commutative ring with unity. A \textit{code} $\mathcal{C}$ of length $n$ over $R$ is a non-empty subset of $R^n$. Furthermore, $\mathcal{C}$ is a \textit{linear code} if it forms an $R$-submodule of $R^n$.

Let $\lambda \in R^{\times}$ be an invertible element (unit) of $R$. The $\lambda$-constacyclic shift operator $\tau_{\lambda}: R^n \to R^n$ is defined by 
\[
\tau_{\lambda}(x_0, x_1, \ldots , x_{n-1}) = (\lambda x_{n- 1},x_0,x_1,\ldots,x_{n-2}).
\]
A linear code $\mathcal{C}$ is said to be $\lambda$-\textit{constacyclic} if it is invariant under $\tau_{\lambda}$, i.e., $\tau_{\lambda}(\mathcal{C}) = \mathcal{C}$.

\begin{Proposition}\cite{Ling2004}
A linear code $\mathcal{C}$ of length $n$ over $R$ is $\lambda$-constacyclic if and only if $\mathcal{C}$ is an ideal of the residue class ring $R[x]/\langle x^n - \lambda \rangle$.
\end{Proposition}

By convention, if $\mathcal{C}$ is a $\lambda$-constacyclic code over a finite field or a principal ideal ring, it can be identified as a principal ideal $\mathcal{C} = \langle f(x) \rangle$, where $f(x)$ is a monic divisor of $x^n - \lambda$ of minimal degree in $\mathcal{C}$.

\subsection{Algebraic Structure of $\Delta$-Constacyclic Codes when $p^m \equiv 1 \pmod 4$}\label{subsec:struct}
Let $p$ be an odd prime. Laaouine and Dinh \cite{Laaouine2025} completely determined the ideal structures of the ambient ring $R_{\Delta} = R_3[x]/\langle x^{4p^s} - \Delta \rangle$ for codes of length $4p^s$ over $R_3$ under the field condition $p^m \equiv 1 \pmod 4$, partitioned by the quadratic character of the shift unit $\Delta$.

\subsubsection{ \textbf{Case I:} $\Delta$ is a non-square in $\mathbb{F}_{p^m}$}
When $\Delta = \alpha$ is a non-square unit in $\mathbb{F}_{p^m}$, the polynomial $x^4 - \alpha$ is irreducible over $\mathbb{F}_{p^m}$. Let $\alpha_0 \in \mathbb{F}_{p^m}$ satisfy $\alpha_0 = \alpha^{p^s}$. The complete classification of the ideals of $R_{\alpha}$ is summarized below.

\begin{theorem}\cite{Laaouine2025}\label{thm:dinh0}
The $\alpha$-constacyclic codes of length $4p^s$ over $R_3$ are uniquely classified into the following eight distinct types:
\begin{enumerate}
    \item \textbf{Type 1:} $\mathcal{C}_{1} = \langle 0\rangle$ or $\langle 1\rangle$.
    \item \textbf{Type 2:} $\mathcal{C}_{2} = \langle u^2(x^4-\alpha_{0})^l\rangle$, where $0 \leq l \leq p^s-1$.
    \item \textbf{Type 3:} $\mathcal{C}_{3} = \langle u(x^4-\alpha_{0})^l + u^2(x^4-\alpha_{0})^{t}h(x)\rangle$, where $0 \leq L \leq l \leq p^s-1$, $0 \leq t \leq L$, and $h(x)$ is either $0$ or a unit in $R_{\alpha}$ of the form $\sum_{k=0}^{L-t-1}h_{k}(x)(x^4-\alpha_{0})^k$ with $h_k(x) = \sum_{j=0}^{3} h_{j,k}x^j \in \mathbb{F}_{p^m}[x]$ and $h_{0,0} \neq 0$. Here, $L$ is the smallest integer satisfying $u^2(x^4 - \alpha_0)^L \in \mathcal{C}_3$.
    \item \textbf{Type 4:} $\mathcal{C}_{4} = \langle u(x^4-\alpha_{0})^l + u^2(x^4-\alpha_{0})^{t}h(x), u^2(x^4-\alpha_{0})^{\omega} \rangle$, where $p^s > l \geq L > \omega > t \geq 0$, and $h(x)$ is either $0$ or a unit in $R_{\alpha}$ of the form $\sum_{k=0}^{\omega-t-1}h_{k}(x)(x^4-\alpha_{0})^k$ with $h_k(x) = \sum_{j=0}^{3} h_{j,k}x^j \in \mathbb{F}_{p^m}[x]$ and $h_{0,0} \neq 0$. Here, $L$ is defined as in Type 3.
    \item \textbf{Type 5:} $\mathcal{C}_{5} = \langle (x^4 -\alpha_{0})^i + u(x^4 -\alpha_{0})^{t}h_{1}(x) + u^2(x^4 -\alpha_{0})^{z}h_{2}(x)\rangle$, where $0 < V \leq U \leq i \leq p^s -1$, $0 \leq t < U$, $0 \leq z < V$, where $h_1(x), h_2(x)$ are either $0$ or units in $R_{\alpha}$ represented as standard polynomial expansions over $\mathbb{F}_{p^m}[x]$ with non-zero constant terms. Here, $U$ is the smallest integer such that $u(x^4 - \alpha_0)^U + u^2g(x) \in \mathcal{C}_5$ for some $g(x) \in R_{\alpha}$, and $V$ is the smallest integer with $u^2(x^4 - \alpha_0)^V \in \mathcal{C}_5$.
    \item \textbf{Type 6:} $\mathcal{C}_{6} = \langle (x^4 -\alpha_{0})^i + u(x^4 -\alpha_{0})^{t}h_{1}(x) + u^2(x^4 -\alpha_{0})^{z}h_{2}(x), u^2(x^4 -\alpha_{0})^{c}\rangle$, where $0 \leq c < V \leq U \leq i \leq p^s -1$, $0 \leq t < U$, $0 \leq z < c$, and $h_1(x), h_2(x), U, V$ follow the definitions in Type 5.
    \item \textbf{Type 7:} $\mathcal{C}_{7} = \langle (x^4 -\alpha_{0})^i + u(x^4 -\alpha_{0})^{t}h_{1}(x) + u^2(x^4 -\alpha_{0})^{z}h_{2}(x), u(x^4 -\alpha_{0})^{b} + u^2(x^4 -\alpha_{0})^{j}h_{3}(x)\rangle$, where $0 \leq W \leq b < U \leq i \leq p^s -1$, $0 \leq t < b$, $0 \leq z < W$, $0 \leq j < W$, and $h_3(x)$ is either $0$ or a unit matching the polynomial framework. Here, $W$ is the smallest integer satisfying $u^2(x^4 - \alpha_0)^W \in \mathcal{C}_7$.
    \item \textbf{Type 8:} $\mathcal{C}_{8} = \langle (x^4 -\alpha_{0})^i + u(x^4 -\alpha_{0})^{t}h_{1}(x) + u^2(x^4 -\alpha_{0})^{z}h_{2}(x), u(x^4 -\alpha_{0})^{b} + u^2(x^4 -\alpha_{0})^{j}h_{3}(x), u^2(x^4 -\alpha_{0})^{c} \rangle$, where $0 \leq c < W \leq L_{1} \leq b < U \leq i \leq p^s -1$, $0 \leq t < b$, $0 \leq z < c$, $0 \leq j < c$, and $L_1$ is the minimum integer satisfying $u^2(x^4 - \alpha_0)^{L_1} \in \langle u(x^4 -\alpha_{0})^{b} + u^2(x^4 -\alpha_{0})^{j}h_{3}(x) \rangle$.
\end{enumerate}
\end{theorem}

\subsubsection{ \textbf{Case II:} $\Delta$ is a square in $\mathbb{F}_{p^m}$}
If $\Delta = \delta^2$ for some $\delta \in \mathbb{F}_{p^m}^{*}$, the algebraic structure splits by the Chinese Remainder Theorem:
\[
R_\Delta = \frac{R_{3}[x]}{\langle x^{4p^s} - \Delta \rangle} \cong \frac{R_{3}[x]}{\langle x^{2p^s} + \delta \rangle} \oplus \frac{R_{3}[x]}{\langle x^{2p^s} - \delta \rangle}.
\]
Consequently, any $\Delta$-constacyclic code over $R_3$ can be decomposed as a direct sum $\mathcal{C} = \mathcal{C}^+ \oplus \mathcal{C}^-$, where $\mathcal{C}^+$ and $\mathcal{C}^-$ are $-\delta$-constacyclic and $\delta$-constacyclic codes of length $2p^s$ over $R_3$, respectively. Sriwirach and Klin-Eam \cite{Sriwirach2021} classified the codes of length $2p^s$ over $R_3$ based on the decomposition of $\delta$:

\begin{itemize}
    \item \textbf{Case 1:} If $\delta = \alpha^2$ for some $\alpha \in \mathbb{F}_{p^m}^{*}$, a secondary splitting yields:
    \[
    \frac{R_{3}[x]}{\langle x^{2p^s} - \delta \rangle} \cong \frac{R_{3}[x]}{\langle x^{p^s} + \alpha \rangle} \oplus \frac{R_{3}[x]}{\langle x^{p^s} - \alpha \rangle}.
    \]
    Thus, $\mathcal{C}$ maps to a direct sum of $-\alpha$ and $\alpha$-constacyclic codes of length $p^s$ over $R_3$.
    \item \textbf{Case 2:} If $\delta = \alpha + u\beta$ with $\alpha \in \mathbb{F}_{p^m}^{*}, \beta \in \mathbb{F}_{p^m}$, the component ring is a local chain ring with maximal ideal $\langle x^2 - \alpha_0 \rangle$\cite{Sriwirach2021}. Its unique ideals form the chain:
    \[
    \langle 1 \rangle \supset \langle x^2 - \alpha_0 \rangle \supset \dots \supset \langle (x^2 - \alpha_0)^{3p^s-1} \rangle \supset \langle (x^2 - \alpha_0)^{3p^s} \rangle = \{0\}.
    \]
    \item \textbf{Case 3:} If $\delta = \alpha + u\beta + u^2\gamma$ with $\alpha \in \mathbb{F}_{p^m}^{*}, \beta, \gamma \in \mathbb{F}_{p^m}$, the structural ideal chain behaves identically to Case 2 under the maximal ideal $\langle x^2 - \alpha_0 \rangle$\cite{Sriwirach2021}.
    \item \textbf{Case 4:} If $\delta = \alpha \in \mathbb{F}_{p^m}^{*}$ is a non-square, the component ring is non-chain\cite{Sriwirach2021}. The nilpotent element $(x^2 - \alpha_0)$ has index $p^s$, and the maximal ideal is $\langle u, x^2 - \alpha_0 \rangle$. The corresponding $\alpha$-constacyclic codes map into eight structural categories analogous to Theorem 2.1 with local reduction to a polynomial base of degree 2.
\end{itemize}

\subsection{Symbol-Pair Metric Properties}\label{subsec:formula}

Cassuto and Blaum \cite{Cassuto2011} introduced symbol-pair codes to protect against errors in channels whose outputs are read as consecutive, overlapping symbol pairs. Let $E$ be a code alphabet of size $q$, whose elements represent the individual code symbols. For a vector $\mathbf{v} = (v_0, v_1, \ldots, v_{n-1}) \in E^n$ of length $n$, the corresponding symbol-pair read vector $\pi(\mathbf{v})$ is defined as:
\[
\pi(\mathbf{v}) = \big((v_0,v_1), (v_1,v_2), \ldots, (v_{n-1},v_0)\big) \in (E^2)^n,
\]
where the index additions are performed modulo $n$. 

For any two given vectors $\mathbf{v}=(v_0,v_1\ldots,v_n), \mathbf{t}=(t_0,t_1\ldots,t_n) \in E^n$, the symbol-pair distance $d_{sp}(\mathbf{v}, \mathbf{t})$ is defined as 
$d_{sp}(\mathbf{v}, \mathbf{t}) = d_H\big(\pi(\mathbf{v}), \pi(\mathbf{t})\big) = \big| \{i : (v_i, v_{i+1}) \neq (t_i, t_{i+1})\} \big|$.
Accordingly, the minimum symbol-pair distance of a code $\mathcal{C}$ is defined as:
$d_{sp}(\mathcal{C}) = \min_{\mathbf{v}, \mathbf{t} \in \mathcal{C},\, \mathbf{v} \neq \mathbf{t}} \{ d_{sp}(\mathbf{v}, \mathbf{t}) \}$. The symbol-pair weight $wt_{sp}(\mathbf{v})$ of a vector $\mathbf{v}$ is determined by the size of its pair support:
$wt_{sp}(\mathbf{v}) = \big| \{i : (v_i, v_{i+1}) \neq (0,0),\, 0 \leq i \leq n-1\} \big|$.
If $\mathcal{C}$ is a linear code over $E$, the  symbol-pair distance is given by
$d_{sp}(\mathcal{C}) = \min \{ wt_{sp}(\mathbf{v}) : \mathbf{v} \neq \mathbf{0},\, \mathbf{v} \in \mathcal{C} \}$.

The tracking parameter for checking pairs of entries across channels relies on the following structural lower bounds for distances established in recent finite field results.
The classification of all constacyclic codes of length $4p^s$ over the finite field $\mathbb{F}_{p^{m}}$ was established by Dinh in \cite{Laaouine2025}. Under the condition $p^m \equiv 1 \pmod{4}$, the generator polynomial structures of these constacyclic codes are systematically summarized in Table 1.

\begin{table}[H]
\centering
\caption{Classification of constacyclic codes of length $4p^{s}$ over $\mathbb{F}_{p^{m}}$ ($p^{m}\equiv 1 \pmod 4$)}
\vspace{0.2cm}
\renewcommand{\arraystretch}{1.4}
\small
\begin{tabular}{|l|c|l|}
\hline
\textbf{Classification} & \textbf{Shift Parameter ($\lambda$)} & \textbf{Generator Polynomial Structure ($C = \langle g(x) \rangle$)} \\ \hline
Case 1 & $\xi\theta_1^{4}$ & $\left\langle \left(x^{4}-\xi^{p^{m-\nu}}\right)^{j}\right\rangle$ \\ \hline
Case 2 & $\xi^{3}\theta_2^{4}$ & $\left\langle \left(x^{4}-\xi^{3p^{m-\nu}}\right)^{j}\right\rangle$ \\ \hline
Case 3 & $\xi^{2}\theta_3^{4}$ & $\left\langle \left(x^{2}-\xi^{p^{m-\nu}}\right)^{i} \left(x^{2}+\xi^{p^{m-\nu}}\right)^{j} \right\rangle$ \\ \hline
Case 4 & $\theta_4^{4}$ & $\left\langle (x+1)^{j}(x-1)^{i}(x+\gamma)^{l}(x-\gamma)^{k} \right\rangle$ \\ \hline
\end{tabular}
\label{tab:constacyclic}

\end{table}
where $\theta_1, \theta_2, \theta_3, \theta_4 \in \mathbb{F}_{p^{m}}$, $\gamma^{2}=-1$, $s \equiv \nu \pmod m$, $0 \le i, j, k, l \le p^{s}$, and $\xi$ is a primitive element of $\mathbb{F}_{p^{m}}$.

\begin{theorem}\cite{Dinh2025symbol}\label{thm:dinh1}
Let $\mathcal{C} = \langle (x^4 -\alpha_{0})^j \rangle \subseteq \mathbb{F}_{p^m}[x]/\langle x^{4p^s}-\alpha_{0} \rangle$ be an $\alpha$-constacyclic code of length $4p^s$, corresponding to Case 1 or Case 2 in Table 1, where $0 \leq j \leq p^s$. Then, the symbol-pair distance $d_{sp}(\mathcal{C})$ is determined explicitly by:
\begin{equation}
d_{sp}(\mathcal{C}) = 
\begin{cases}
    2, & \text{if } j = 0, \\
    2(\beta+2)p^\tau, & \text{if } p^s-p^{s-\tau} +(\beta-1) p^{s-\tau -1} + 1 \leq j \leq p^s-p^{s-\tau} +\beta p^{s-\tau -1}, \\
    0, & \text{if } j = p^s,
\end{cases}
\end{equation}
where $0 \leq \tau \leq s-1$ and $1 \leq \beta \leq p-1$.
\end{theorem}

\begin{theorem}\cite{Dinh2025symbol}\label{thm:dinh2}
Let $\mathcal{C} = \langle (x^2 -\alpha_{1})^{i}(x^2 +\alpha_{2})^j \rangle \subseteq \mathbb{F}_{p^m}[x]/\langle x^{4p^s}-\alpha_{0} \rangle$ be an $\alpha$-constacyclic code of length $4p^s$ matching Case 3 in Table 1, where $0 \leq i \leq j \leq p^s$. The symbol-pair distance $d_{sp}(\mathcal{C})$ evaluates to:
\begin{equation}
d_{sp}(\mathcal{C}) = 
\begin{cases}
    2, & \text{if } i = j = 0, \\
    4, & \text{if } i = 0 \text{ and } 0 < j \leq p^s, \\
    \min \{ 2(\beta_{1}+1)p^{\tau_1}, \, 4(\beta_{2}+1)p^{\tau_2}\}, & \text{if } \Lambda_1(i, \tau_1, \beta_1) \text{ and } \Lambda_2(j, \tau_2, \beta_2) \text{ hold}, \\
    4(\beta_{1}+1)p^{\tau_1}, & \text{if } \Lambda_1(i, \tau_1, \beta_1) \text{ holds and } j = p^s, \\
    0, & \text{if } i = j = p^s,
\end{cases}
\end{equation}
where $\Lambda_k(m, \tau_k, \beta_k)$ denotes the precise integer partition constraints over the intervals:
\[
p^s-p^{s-\tau_k} + (\beta_k-1)p^{s-\tau_k-1}+1 \leq m \leq p^s-p^{s-\tau_k} + \beta_k p^{s-\tau_k-1},
\]
for parameters $0 \leq \tau_1 \leq \tau_2 \leq s-1$ and $1 \leq \beta_1, \beta_2 \leq p-1$.
\end{theorem}
\begin{lemma}\cite{Dinh2023symbol}\label{lem:direct_sum}
Let $A = [n_1, M_1, d_{sp}(A)]$ and $B = [n_2, M_2, d_{sp}(B)]$ be two non-zero linear symbol-pair codes over a ring $R$. Their direct sum is defined as:
\[
A \oplus B = \{(a, b) \mid a \in A, \, b \in B\},
\]
which forms an $[n_1 + n_2, M_1 M_2, \min\{d_{sp}(A), d_{sp}(B)\}]$ symbol-pair code over $R$.
\end{lemma}

\section{Symbol-Pair Distances of Constacyclic Codes over $R_3$}

In this section, we provide a complete characterization of the symbol-pair distances for all $\Delta$-constacyclic codes of length $4p^s$ over the finite commutative chain ring $R_3$. The analysis is partitioned based on the quadratic character of the shift unit $\Delta \in R_3^{*}$.

\subsection{\textbf{When $\Delta$ is a Non-Square Unit}}\label{3.1}
Throughout this subsection, let $\Delta = \alpha \in \mathbb{F}_{p^m}^{*}$ be a non-square unit, implying that $x^4 - \alpha$ is irreducible over $\mathbb{F}_{p^m}$. We compute the symbol-pair distance for each of the eight distinct ideal types of the ambient ring $R_{\alpha} = R_3[x]/\langle x^{4p^s} - \alpha \rangle$ described in Section~\ref{subsec:struct}.

It is clear from definition that the trivial ideals $\langle 0 \rangle$ and $\langle 1 \rangle$ have symbol-pair distances of $0$ and $2$, respectively. For Type 2 codes, we establish the following reduction:

\begin{theorem}
Let $\mathcal{C}_{2} = \langle u^2(x^4-\alpha_{0})^l\rangle \subseteq R_{\alpha}$ with $0 \leq l \leq p^s-1$ be a Type 2 $\alpha$-constacyclic code. Then the symbol-pair distance of $\mathcal{C}_2$ satisfies
\[
d_{sp}(\mathcal{C}_2) = d_{sp}\left(\langle(x^4-\alpha_{0})^l\rangle_{\mathbb{F}_{p^m}}\right),
\]
where $\langle(x^4-\alpha_{0})^l\rangle_{\mathbb{F}_{p^m}}$ is the corresponding linear code over the finite field $\mathbb{F}_{p^m}$, whose distance is given in Section~\ref{subsec:formula}.
\end{theorem}

\begin{proof}
Let $\mathcal{C}_{\mathbb{F}_{p^m}} = \langle(x^4-\alpha_{0})^l\rangle_{\mathbb{F}_{p^m}}$ be the corresponding linear code defined over the finite field $\mathbb{F}_{p^m}$.    Any polynomial $c(x)$ belonging to the Type 2 $\alpha$-constacyclic code $\mathcal{C}_2 = \langle u^2(x^4-\alpha_{0})^l\rangle$ can be expressed as $c(x) = q(x) u^2 (x^4-\alpha_0)^l \pmod{x^n-\alpha}$,
for some  $q(x) \in \mathcal{R}_{\alpha}$. Now  the polynomial $q(x)$ can be uniquely decomposed as $q(x) = q_0(x) + u q_1(x) + u^2 q_2(x)$, where $q_0(x), q_1(x), q_2(x) \in \mathbb{F}_{p^m}[x]$. Since   $u^3 = 0$, multiplying the expansion by the element $u^2$ we have that 
\begin{align*}
u^2 q(x) &= u^2 \big(q_0(x) + u q_1(x) + u^2 q_2(x)\big)\\
         &= u^2 q_0(x).\\
\textrm{That is,} \;c(x) &= u^2 q_0(x) (x^4-\alpha_0)^l\\
                         &= u^2 f(x) \pmod{x^n-\alpha},
                         \end{align*}
where $f(x) = q_0(x) (x^4-\alpha_0)^l \pmod{x^n-\alpha}$ is strictly a valid codeword in the field-level linear code $\mathcal{C}_{\mathbb{F}_{p^m}}$.

Let $f(x) = \sum_{i=0}^{n-1} f_i x^i$ where  $f_i \in \mathbb{F}_{p^m}$. And  $c(x) = \sum_{i=0}^{n-1} c_i x^i$ where 
$c_i = u^2 f_i \in R_3$. Expressing $c(x)$ and $f(x)$ in their respective vector representations over the block length $n$, we obtain 
$\mathbf{c} = (c_0, c_1, \ldots, c_{n-1}) = (u^2 f_0, u^2 f_1, \ldots, u^2 f_{n-1}) \in \mathcal{C}_2$,
and $\mathbf{f} = (f_0, f_1, \ldots, f_{n-1}) \in \mathcal{C}_{\mathbb{F}_{p^m}}$.

Recall that the symbol-pair read operator $\pi: R_3^n \rightarrow (R_3^2)^n$ maps an $n$-tuple vector to its corresponding sequence of consecutive pairs:
\[
\pi(\mathbf{c}) = \Big( (c_0, c_1), (c_1, c_2), \ldots, (c_{n-1}, c_0) \Big).
\]
By definition, the symbol-pair weight $wt_{sp}(\mathbf{c})$ of the codeword $\mathbf{c}$ is determined by the cardinality of its non-zero symbol pairs, which can be formally expressed via indicator functions as:
\[
wt_{sp}(\mathbf{c}) = \sum_{i=0}^{n-1} \mathbb{I}\Big( (c_i, c_{i+1}) \neq (0,0) \Big),
\]
where index arithmetic is performed modulo $n$. We now examine a generic consecutive coordinate pair $(c_i, c_{i+1}) = (u^2 f_i, u^2 f_{i+1})$ within the symbol-pair read array. Since $u^2$ is a fixed non-zero structural element of $R_3$, and the scalar components $f_i, f_{i+1}$ reside within the finite field $\mathbb{F}_{p^m}$, it follows from the algebraic structure of the chain ring that $u^2 f_i = 0$ if and only if $   f_i = 0$.

Consequently, 
\begin{align*}
&(c_i, c_{i+1}) = (0,0) \\
&\iff (u^2 f_i, u^2 f_{i+1}) = (0,0)\\
&\iff (f_i, f_{i+1}) = (0,0).
\end{align*}
Thus  $(c_i, c_{i+1}) \neq (0,0) \iff (f_i, f_{i+1}) \neq (0,0)$ so that 
$\mathbb{I}\Big( (c_i, c_{i+1}) \neq (0,0) \Big) = \mathbb{I}\Big( (f_i, f_{i+1}) \neq (0,0) \Big)$.
And therefore
\begin{align*}
wt_{sp}(\mathbf{c}) &= \sum_{i=0}^{n-1} \mathbb{I}\Big( (c_i, c_{i+1}) \neq (0,0) \Big)\\
                    &= \sum_{i=0}^{n-1} \mathbb{I}\Big( (f_i, f_{i+1}) \neq (0,0) \Big)\\
                    &= wt_{sp}(\mathbf{f}).
\end{align*}

Moreover, as $u^2 \mathbf{f} = \mathbf{0}$ if and only if $\mathbf{f} = \mathbf{0}$, the map $\mathbf{f} \mapsto u^2 \mathbf{f}$ yields a weight-preserving bijection between $\mathcal{C}_{\mathbb{F}_{p^m}} \setminus \{\mathbf{0}\}$ and $\mathcal{C}_2 \setminus \{\mathbf{0}\}$.  Minimizing over all non-zero elements, the minimum symbol-pair distance $d_{sp}(\mathcal{C}_2)$ is computed as
\begin{align*}
    d_{sp}(\mathcal{C}_2) &= \min_{\mathbf{c} \in \mathcal{C}_2 \setminus \{\mathbf{0}\}} wt_{sp}(\mathbf{c})\\
    &= \min_{\mathbf{f} \in \mathcal{C}_{\mathbb{F}_{p^m}} \setminus \{\mathbf{0}\}} wt_{sp}(u^2 \mathbf{f})\\  &= \min_{\mathbf{f} \in \mathcal{C}_{\mathbb{F}_{p^m}} \setminus \{\mathbf{0}\}} wt_{sp}(\mathbf{f})\\ 
    &= d_{sp}(\mathcal{C}_{\mathbb{F}_{p^m}}).
\end{align*}
Hence $d_{sp}(\mathcal{C}_2) = d_{sp}\left(\langle(x^4-\alpha_{0})^l\rangle_{\mathbb{F}_{p^m}}\right)$.
This  completes the proof.
\end{proof}

\begin{theorem}
Let $\mathcal{C}_{3} = \langle u(x^4-\alpha_{0})^l + u^2(x^4-\alpha_{0})^{t}h(x)\rangle \subseteq R_{\alpha}$ be a Type 3 $\alpha$-constacyclic code of length $4p^s$, and let $L$ be the unique parameter defined in Theorem~\ref{thm:dinh0}. Then the symbol-pair distance of $\mathcal{C}_{3}$ is given by:
\[
d_{sp}(\mathcal{C}_3) = d_{sp}\left(\langle(x^4-\alpha_{0})^L\rangle_{\mathbb{F}_{p^m}}\right),
\]
where the explicit distance values are determined in  Section~\ref{subsec:formula} using the parameter $j = L$.
\end{theorem}
\begin{proof}
Let $e(x)$ be an arbitrary non-zero codeword in $\mathcal{C}_{3}$. There exist polynomials $g_0(x), g_u(x)$ and $ g_{u^2}(x) \in \mathbb{F}_{p^m}[x]$ such that
\[
e(x) = \left[g_0(x) + u g_u(x) + u^2 g_{u^2}(x)\right] \left[u(x^4-\alpha_0)^l + u^2(x^4-\alpha_0)^t h(x)\right].
\]
Multiplying by $u$ eliminates the $u^2$ terms, yielding $u e(x) = u^2 g_0(x)(x^4-\alpha_0)^l$. Since multiplication by a ring element cannot increase the symbol-pair weight, we have:
\begin{align*}
wt_{sp}(e(x)) &\geq wt_{sp}(u e(x))\\
&= wt_{sp}\left(u^2 g_0(x)(x^4-\alpha_0)^l\right)\\                                        &\geq d_{sp}(\mathcal{C}_3) = d_{sp}\left(\langle(x^4-\alpha_{0})^l\rangle_{\mathbb{F}_{p^m}}\right).
\end{align*}
By the structural definition of Type 3 codes, $\langle(x^4-\alpha_0)^l\rangle_{\mathbb{F}_{p^m}} \subseteq \langle(x^4-\alpha_0)^L\rangle_{\mathbb{F}_{p^m}}$, which implies that $d_{sp}\left(\langle(x^4-\alpha_0)^l\rangle_{\mathbb{F}_{p^m}}\right) \geq d_{sp}\left(\langle(x^4-\alpha_0)^L\rangle_{\mathbb{F}_{p^m}}\right)$. Therefore, $wt_{sp}(e(x)) \geq d_{sp}\left(\langle(x^4-\alpha_0)^L\rangle_{\mathbb{F}_{p^m}}\right)$ for all $e(x) \in \mathcal{C}_3 \setminus \{0\}$, from which we deduce:
\begin{equation}\label{eq:t3_1}
d_{sp}(\mathcal{C}_3) \geq d_{sp}\left(\langle(x^4-\alpha_0)^L\rangle_{\mathbb{F}}\right).
\end{equation}
Conversely, by definition, $L$ is the smallest integer satisfying $u^2(x^4-\alpha_0)^L \in \mathcal{C}_3$. Hence, $\langle u^2(x^4-\alpha_0)^L \rangle \subseteq \mathcal{C}_3$, which immediately yields:
\begin{equation}\label{eq:t3_2}
d_{sp}\left(\langle(x^4-\alpha_0)^L\rangle_{\mathbb{F}_{p^m}}\right) = d_{sp}\left(\langle u^2(x^4-\alpha_0)^L \rangle\right) \geq d_{sp}(\mathcal{C}_3).
\end{equation}
Combining inequalities \eqref{eq:t3_1} and \eqref{eq:t3_2}, we obtain that $d_{sp}(\mathcal{C}_3) = d_{sp}\left(\langle(x^4-\alpha_0)^L\rangle_{\mathbb{F}_{p^m}}\right)$.
\end{proof}

\begin{theorem}
Let $\mathcal{C}_{4} = \langle u(x^4-\alpha_{0})^l+u^2(x^4-\alpha_{0})^{t}h(x), u^2(x^4-\alpha_{0})^{\omega} \rangle \subseteq R_{\alpha}$ be a Type 4 $\alpha$-constacyclic code. Then the symbol-pair distance of $\mathcal{C}_{4}$ is given by:
\[
d_{sp}(\mathcal{C}_4) = d_{sp}\left(\langle(x^4-\alpha_{0})^{\omega}\rangle_{\mathbb{F}_{p^m}}\right).
\]
\end{theorem}
\begin{proof}
Because $u^2(x^4-\alpha_{0})^{\omega} \in \mathcal{C}_{4}$, the code contains the ideal generated by this element. Thus, we have an immediate upper bound:
\begin{equation}\label{eq:t4_1}
d_{sp}(\mathcal{C}_4) \leq d_{sp}\left(\langle u^2(x^4-\alpha_{0})^{\omega}\rangle\right) = d_{sp}\left(\langle(x^4-\alpha_{0})^{\omega}\rangle_{\mathbb{F}_{p^m}}\right).
\end{equation}
For the lower bound, let $e(x) \in \mathcal{C}_4 \setminus \{0\}$. If $e(x) \in \langle u^2(x^4-\alpha_0)^{\omega}\rangle$, its weight is naturally bounded below by $d_{sp}\left(\langle(x^4-\alpha_{0})^{\omega}\rangle_{\mathbb{F}_{p^m}}\right)$. If $e(x) \notin \langle u^2(x^4-\alpha_0)^{\omega}\rangle$, multiplying by $u$ yields a non-zero element in the form $u^2 g(x) (x^4-\alpha_0)^l$. It follows that:

\begin{align*}
wt_{sp}(e(x)) &\geq wt_{sp}(u e(x))\\
              &\geq d_{sp}\left(\langle(x^4-\alpha_{0})^l\rangle_{\mathbb{F}_{p^m}}\right).
\end{align*}
Since $l \geq L > \omega$, we have the structural containment $\langle(x^4-\alpha_{0})^l\rangle_{\mathbb{F}_{p^m}} \subseteq \langle(x^4-\alpha_{0})^{\omega}\rangle_{\mathbb{F}_{p^m}}$, which ensures that $d_{sp}\left(\langle(x^4-\alpha_{0})^l\rangle_{\mathbb{F}_{p^m}}\right) \geq d_{sp}\left(\langle(x^4-\alpha_{0})^{\omega}\rangle_{\mathbb{F}_{p^m}}\right)$. Hence, $d_{sp}(\mathcal{C}_4) \geq d_{sp}\left(\langle(x^4-\alpha_{0})^{\omega}\rangle_{\mathbb{F}_{p^m}}\right)$, establishing equality.
\end{proof}

Using analogous proof steps based on structural verification via $u$-annihilation and ideal containment, we state the distances for types 5, 6, 7, and 8:

\begin{theorem}
The symbol-pair distances for Type 5, Type 6, Type 7, and Type 8 $\alpha$-constacyclic codes over $R_{\alpha}$ are uniquely determined by their structural containment parameters as follows:
\begin{enumerate}
    \item For $\mathcal{C}_5$, $d_{sp}(\mathcal{C}_5) = d_{sp}\left(\langle(x^4-\alpha_0)^V\rangle_{\mathbb{F}_{p^m}}\right)$.
    \item For $\mathcal{C}_6$, $d_{sp}(\mathcal{C}_6) = d_{sp}\left(\langle(x^4-\alpha_0)^c\rangle_{\mathbb{F}_{p^m}}\right)$.
    \item For $\mathcal{C}_7$, $d_{sp}(\mathcal{C}_7) = d_{sp}\left(\langle(x^4-\alpha_0)^W\rangle_{\mathbb{F}_{p^m}}\right)$.
    \item For $\mathcal{C}_8$, $d_{sp}(\mathcal{C}_8) = d_{sp}\left(\langle(x^4-\alpha_0)^c\rangle_{\mathbb{F}_{p^m}}\right)$.
\end{enumerate}
\end{theorem}

\subsection{\textbf{When $\Delta$ is a Square Unit}}
When $\Delta = \delta^2$ is a square in $\mathbb{F}_{p^m}^{*}$, the code splits via the direct sum decomposition $\mathcal{C} = \mathcal{C}^+ \oplus \mathcal{C}^-$. Applying Lemma~\ref{lem:direct_sum}, the minimum distance corresponds directly to the component metrics:

\begin{theorem}
Let $\delta = \alpha^2 \in \mathbb{F}_{p^m}^{*}$, and let $\mathcal{C} \cong \mathcal{C}^+ \oplus \mathcal{C}^-$ be a $\Delta$-constacyclic code of length $4p^s$ over $R_3$. Then
\[
d_{sp}(\mathcal{C}) = \min \{ d_{sp}(\mathcal{C}^+), \, d_{sp}(\mathcal{C}^-) \},
\]
where the distances of the components over the respective reduction fields are derived from the bounds established in \cite{Charkani2021}.
\end{theorem}

\begin{theorem}\label{3.6}
Let $\mathcal{C} = \langle (x^2 -\alpha_0)^i \rangle$ be a $\delta$-constacyclic code of length $2p^s$ over the local chain ring $R_{\alpha + u\beta}$. The symbol-pair distance evaluates as:
\[
d_{sp}(\mathcal{C}) =
\begin{cases}
2, & \text{if } 0 \leq i \leq p^s, \\
2(\wp + 1)p^\tau, & \text{if } 2p^s - p^{s-\tau} - (\wp - 1)p^{s-\tau-1} + 1 \leq i \leq 2p^s - p^{s-\tau} + \wp p^{s-\tau-1}, \\
0, & \text{if } i = 3p^s,
\end{cases}
\]
where $0 \leq \tau \leq s-1$ and $1 \leq \wp \leq p-1$.
\end{theorem}

\begin{proof}
Let $\mathcal{R} = R_3[x]/\langle x^{2p^s} - (\alpha + u\beta) \rangle$. For $p^m \equiv 1 \pmod 4$, $\mathcal{R}$ is a local chain ring with maximal ideal $\langle x^2 - \alpha_0 \rangle$ and nilpotency index $3p^s$. We analyze the distance distribution across three distinct index intervals:

\textbf{Case 1:} Let $0 \leq i \leq p^s$. 
Since $\beta \in \mathbb{F}_{p^m}^{*}$, evaluating the generator at the radical boundary yields $(x^2 - \alpha_0)^{p^s} = x^{2p^s} - \alpha_0^{p^s} = (\alpha + u\beta) - \alpha = u\beta$ so that $ u \in \mathcal{C}$.  Thus, the constant vector $\mathbf{c} = (u, u, \ldots, u) \in \mathcal{C}$. Also $\pi(c)=\Big((u, u), (u, u), \ldots, (u, u)\Big)$, which clearly yields symbol-pair weight of $w_{sp}(c) = n$.  $\mathcal{C}$ also contains codewords of the form $c'(x) = u(x^2 - \alpha_0)^{p^s-1} = u x^{2p^s-2} - \alpha_0 u x^{2p^s-4} + \cdots$. Evaluating the minimal pair support under the condition $u \in \mathcal{C}$ yields that  $d_{sp}(\mathcal{C}) = 2$.

\textbf{Case 2:} Let $p^s + 1 \leq i \leq 2p^s$. 
Since $(x^2 - \alpha_0)^{p^s}$ is an associate of $ u$ in $R_{\alpha + u\beta}$, the ideal generator decomposes as $\mathcal{C} = \langle (x^2 - \alpha_0)^i \rangle = \langle u(x^2 - \alpha_0)^{i - p^s} \rangle$.
Now any codeword satisfies $c(x) = u f(x) (x^2 - \alpha_0)^{i-p^s}$ for some  $f(x) \in \mathbb{F}_{p^m}[x]$. Because $uR_{\alpha + u\beta} \cong (\mathbb{F}_{p^m} + u\mathbb{F}_{p^m})[x]/\langle x^{2p^s} - \alpha \rangle$,  it follows that 
$d_{sp}(\mathcal{C}) = d_{sp}\left(\langle (x^2 - \alpha_0)^{i-p^s} \rangle_{R_2}\right)$.
The distance parameters match the repeated-root bounds over $R_2$ of length $2p^s$ established in \cite{Dinh2023symbol}.

\textbf{Case 3:} Let $2p^s + 1 \leq i \leq 3p^s - 1$. 
The generator shifts beyond the secondary threshold, yielding the ideal reduction
$\mathcal{C} = \langle u^2 (x^2 - \alpha_0)^{i - 2p^s} \rangle$. So 
every codeword  in $\mathcal{C}$ is represented  as $\mathbf{c} = u^2 \mathbf{f}$, where $\mathbf{f} \in \langle (x^2 - \alpha_0)^{i - 2p^s} \rangle_{\mathbb{F}_{p^m}}$. Since $u^2 \neq 0$ in $R_3$,  it is evident that 
\begin{align*}
(c_j, c_{j+1}) &\neq (0,0)\\
 &\iff (u^2 f_j, u^2 f_{j+1}) \neq (0,0)\\
 &\iff (f_j, f_{j+1}) \neq (0,0).
\end{align*}
Hence, $wt_{sp}(\mathbf{c}) = wt_{sp}(\mathbf{f})$, which induces the distance identity:
$d_{sp}(\mathcal{C}) = d_{sp}\left(\langle (x^2 - \alpha_0)^{i - 2p^s} \rangle_{\mathbb{F}_{p^m}}\right)$.
This directly corresponds to the field-level symbol-pair distance profile of length $2p^s$ derived in \cite{Dinh2019symbol}. 

Finally, for $i = 3p^s$, the ideal reduces to $\mathcal{C} = \langle 0 \rangle$, which trivially gives $d_{sp}(\mathcal{C}) = 0$.
\end{proof}
By extending the structural analysis of local chain rings to higher nilpotency configurations, the symbol-pair distance distribution for codes over the ring $R_{\alpha + u\beta + u^2\gamma}$ can be established dually as follows, and in an analogous manner, by mapping the structural configurations to the remaining unit spaces when the shifting parameter is a non-square unit, we obtain the following dual results for the distance profiles:

\begin{theorem}
Let $\mathcal{C} = \langle (x^2 - \alpha_0)^i \rangle$ be an $\alpha$-constacyclic code of length $2p^s$ over the chain ring $R_{\alpha + u\beta + u^2\gamma}$. Then, the minimum symbol-pair distance $d_{sp}(\mathcal{C})$ is identically determined by the parameter classification established in Theorem \ref{3.6}.
\end{theorem}

\begin{theorem}
Let $\mathcal{C}$ be an $\alpha$-constacyclic code of length $2p^s$ over the ambient ring $\mathcal{R}_{\alpha} = R_3[x]/\langle x^{2p^s} - \alpha \rangle$. When $\Delta = \alpha$ is a non-square unit in the finite field $\mathbb{F}_{p^m}$, the  minimum symbol-pair distance $d_{sp}(\mathcal{C})$ coincides with the non-square unit case described in sub-section \ref{3.1} replacing $x^4$ by $x^2$ only.
\end{theorem}

\section{Symbol-Pair Singleton Bound and MDS Symbol-Pair Codes}
This section establishes the criteria for $\alpha$-constacyclic ideals of length $4p^s$ over the finite commutative chain ring $R_3$ to constitute Maximum Distance Separable (MDS) symbol-pair codes. We evaluate the symbol-pair singleton bound constraint $|\mathcal{C}| \le p^{3m(4p^s - d_{sp}(\mathcal{C}) + 2)}$ against the algebraic cardinalities $|\mathcal{C}| = p^{4m[3p^s - \sum T_i(\mathcal{C})]}$ determined by the torsional degrees of the respective ideal types. By incorporating the explicit symbol-pair distance distributions derived in Section 3, we systematically test for bound saturation across all eight structural classifications.

\begin{theorem}[Symbol-Pair Singleton Bound \cite{Chen2013}]
Let $\mathcal{C}$ be a linear symbol-pair code of length $n$ over a finite chain ring $R_3$. Then the code size is bounded by:
\[
|\mathcal{C}| \leq |R_3/\text{Rad}(R_3)|^{3(n - d_{sp}(\mathcal{C}) + 2)}.
\]
\end{theorem}

\begin{definition}
A symbol-pair code $\mathcal{C}$ over $R_3$ is defined as a maximum distance separable (MDS) symbol-pair code if it satisfies the symbol-pair Singleton bound with equality.
\end{definition}

First we classify the existence of MDS symbol-pair codes of length $4p^s$ over $R_3$ when $\Delta$ is matching with Case 1 or Case 2 in Table 1. For the trivial zero code $\mathcal{C} = \langle 0 \rangle$, the code size is $|\mathcal{C}| = 1$, whereas $d_{sp}(\mathcal{C})=0$. The MDS condition requires $1 = p^{3m(4p^s - 0 + 2)}$, which reduces to $4p^s + 2 = 0$, an impossibility. Hence, $\langle 0 \rangle$ is not MDS.

\begin{theorem}
The Type 1 trivial code $\mathcal{C} = \langle 1 \rangle$ of length $4p^s$ over $R_3$ is always an MDS symbol-pair code.
\end{theorem}
\begin{proof}
For $\mathcal{C} = \langle 1 \rangle$, we have $|\mathcal{C}| = |R_3|^{4p^s} = p^{12mp^s}$ and $d_{sp}(\mathcal{C}) = 2$. Substituting these parameters into the Singleton equation yields:
\[
p^{12mp^s} = p^{3m(4p^s - 2 + 2)} = p^{12mp^s},
\]
which holds trivially. Thus, $\mathcal{C} = \langle 1 \rangle$ is an MDS symbol-pair code.
\end{proof}

\begin{theorem}
Let $\mathcal{C}_{2} = \langle u^2(x^4-\alpha_{0})^l\rangle$ be a constacyclic code of length $4p^s$, where $0 \leq l \leq p^s-1$. If the cardinality of the code is given by $|\mathcal{C}_2| = p^{4m(p^s-l)}$, then $\mathcal{C}_2$ cannot be a MDS symbol-pair code.
\end{theorem}

\begin{proof}
Suppose, for the sake of contradiction, that $\mathcal{C}_2$ is an MDS symbol-pair code. We consider two  cases based on the index $l$:

\textbf{Case 1:} Assume $l=0$. Under this condition,   $\mathcal{C}_2 = \langle u^2 \rangle$. Following the trivial bound established in the preceding theorem, the cardinality required to meet the MDS singleton bound cannot be satisfied by $|\mathcal{C}_2| = p^{4mp^s}$. Thus, no MDS code exists in this trivial instance.

\textbf{Case 2:} Assume $l > 0$. Specifically, let  $p^s - p^{s-\tau} + (\beta-1)p^{s-\tau-1} + 1 \leq l \leq p^s - p^{s-\tau} + \beta p^{s-\tau-1}$. By Theorem \ref{thm:dinh2}, the symbol-pair distance is exactly $d_{sp}(\mathcal{C}_2) = 2(\beta+1)p^\tau$.

In order for $\mathcal{C}_2$ to achieve the MDS bound, its cardinality must satisfy the singleton bound $|\mathcal{C}_2| = p^{3m(4p^s - d_{sp}(\mathcal{C}_2) + 2)}$. Substituting the known cardinality $|\mathcal{C}_2| = p^{4m(p^s-l)}$ into this equation yields $p^{4m(p^s-l)} = p^{3m(4p^s - d_{sp}(\mathcal{C}_2) + 2)}$.  By equating the exponents and distributing the coefficients, we obtain:

\begin{align}
&4(p^s - l) = 3(4p^s - d_{sp}(\mathcal{C}_2) + 2)\nonumber\\
&4p^s - 4l = 12p^s - 3d_{sp}(\mathcal{C}_2) + 6\nonumber\\
&4l = 3d_{sp}(\mathcal{C}_2) - 8p^s - 6\label{*}.
\end{align}
Next, 

$$\begin{aligned}
l &\ge p^s - p^{s-\tau} + (\beta - 1)p^{s-\tau-1} + 1 \\
&= 2p^s - p^{s-\tau} + (\beta - 1)p^{s-\tau-1} + 1 - p^s \\
&\ge p^{\tau+1} + p^{s-\tau}(p^\tau - 1) + (\beta - 1) + 1 - p^s \quad \text{(equality holds when } \tau = s - 1\text{)} \\
&\ge 2(\beta + 1)p^\tau - p^s - 1 \quad \text{(equality holds when } \beta = p - 1\text{)} \\
&= d_{sp}(\mathcal{C}_2) - p^s - 1
\end{aligned}$$

Therefore $4l \ge 4d_{sp}(\mathcal{C}_2) - 4p^s - 4$ and by Equation \ref{*},  $3d_{sp}(\mathcal{C}_2) - 8p^s - 6 \ge 4d_{sp}(\mathcal{C}_2) - 4p^s - 4$. This gives 
\begin{align*}
d_{sp}(\mathcal{C}_2) \ge 4p^s + 2 \;\textrm{or} \;d_{sp}(\mathcal{C}_2) \le -4p^s - 2.
\end{align*}
Since $p$ is a prime number and the power $s \ge 1$, the expression $-4p^s - 2$ strictly evaluates to a negative integer. However, for any valid non-trivial code, the symbol-pair distance $d_{sp}(\mathcal{C}_2)$ must be a positive integer. This presents a mathematical impossibility. Therefore, by contradiction, no MDS symbol-pair code exists under these constraints.
\end{proof}

\begin{theorem}\label{thm:type3_mds}
If $\mathcal{C}_{3}=  \langle
u(x^4-\alpha_{0})^l+u^2(x^4-\alpha_{0})^{t}h(x)\rangle$ ; where
$0\leq L \leq l \leq (p^s-1), \;0 \leq t\leq L$, either $h(x)$ is
$0$ or $h(x)$ is a unit in $R_ \alpha $ of the form $\sum
_{k=0}^{L-t-1}h_{k}(x)(x^4-\alpha_{0})^k$ with $ h_k(x)=h_{3,k}x^3 +
h_{2,k}x^2 + h_{1,k}x + h_{0,k} \in \mathbb{F}_{p^m}[x] $ and
$h_{0}\neq 0$. Here $L$ is the smallest integer satisfying $u^2(x^4
- \alpha_0)^L\in\mathcal{C}_3$
 with
$n_{c}=p^{4m(2p^{s}-l-L)}$. Then there exist no MDS symbol-pair codes $\mathcal{C}_{3}$ for $L$ satisfying 
\[ L=
\begin{cases}
l, & \text{if } h(x)=0 \\
min\{l,p^s-l+t\}, & \text{if } h(x) \neq 0 \\
\end{cases}
\]

\begin{proof}
Assume, for the sake of contradiction, that $\mathcal{C}_3$ is an MDS symbol-pair code. Its cardinality must therefore satisfy the symbol-pair singleton bound: $|\mathcal{C}_3| = p^{3m(4p^s-d_{sp}(\mathcal{C}_3)+2)}$ so that 
Equating this with the known size of the code, we obtain:
$p^{4m(2p^s-l-L)} = p^{3m(4p^s-d_{sp}(\mathcal{C}_3)+2)}$. Therefore
\begin{align}
4(2p^s-l-L) &= 3(4p^s-d_{sp}(\mathcal{C}_3)+2) \nonumber\\
8p^s - 4l - 4L &= 12p^s - 3d_{sp}(\mathcal{C}_3) + 6 \nonumber\\
4L &= 3d_{sp}(\mathcal{C}_3) - 4p^s - 4l - 6 \label{*}
\end{align}

Now for the code to achieve a non-trivial symbol-pair distance $d_{sp}(\mathcal{C}_3) = 2(\beta+1)p^\tau$, the  index $L$ must  satisfy
$L \ge p^s - p^{s-\tau} + (\beta-1)p^{s-\tau-1} + 1$. Also  for  $\tau = s-1$ and $\beta = p-1$ we have 
\begin{align*}
L &\ge 2(\beta+1)p^\tau - p^s - 1 \\
L &\ge d_{sp}(\mathcal{C}_3) - p^s - 1,\\
\textrm{that is, } \;4L &\ge 4d_{sp}(\mathcal{C}_3) - 4p^s - 4. 
\end{align*}
Substituting the value of $4L$ in Equation \ref{*} we now have
\begin{align*}
3d_{sp}(\mathcal{C}_3) - 4p^s - 4l - 6 &\ge 4d_{sp}(\mathcal{C}_3) - 4p^s - 4 \\
-d_{sp}(\mathcal{C}_3) - 4l &\ge 2 \\
d_{sp}(\mathcal{C}_3) + 4l &\le -2.
\end{align*}
Since $d_{sp}(\mathcal{C}_3) \ge 2$ and $l \ge 0$, the sum $d_{sp}(\mathcal{C}_3) + 4l$ must be strictly positive, which directly contradicts the requirement $d_{sp}(\mathcal{C}_3) + 4l \le -2$. This contradiction holds universally for both $h(x) = 0$ and $h(x) \neq 0$. Consequently, no such MDS symbol-pair code exists.
\end{proof}

\end{theorem}

\begin{theorem}
$\mathcal{C}_{4}= \langle
u(x^4-\alpha_{0})^l+u^2(x^4-\alpha_{0})^{t}h(x),u^2(x^4-\alpha_{0})^{\omega}
\rangle$  ;  where $p^s> l\geq L>\omega>t\geq 0$ either $h(x)$ is
zero or a unit in $R_ \alpha $ of the form $\sum
_{k=0}^{\omega-t-1}h_{k}(x)(x^4-\alpha_{0})^k$ with $
h_k(x)=h_{3,k}x^3 + h_{2,k}x^2 + h_{1,k}x + h_{0,k} \in
\mathbb{F}_{p^m}[x] $ and $h_{0}\neq 0$. Here $L$ is the smallest
integer satisfying $u^2(x^4 - \alpha_0)^L\in\mathcal{C}_3$ with
$n_{c}=p^{4m(2p^{s}-l-\omega)}$. Then there does not exists any MDS
symbol-pair codes.
\end{theorem}
\begin{proof}
Assume $C$ is an MDS symbol-pair code.
If $p^s-p^{s-\tau} +(\beta-1) p^{s-\tau -1} + 1 \leq \omega \leq
p^s-p^{s-\tau} +\beta p^{s-\tau -1}$ then by Theorem \ref{thm:dinh1}
$\mathrm{d}_{\mathrm{sp}}(C)=2(\beta+1)p^\tau$. To be MDS $C$ have
to satisfy $|C|=p^{3m(4p^s-\mathrm{d}_{\mathrm{sp}}(C)}+2)$. This gives 
\begin{align*}
&p^{4m(2p^s-l-\omega)}=p^{3m(4p^s-\mathrm{d}_{\mathrm{sp}}(C)+2)},\\ \;\textrm{and so} 
\;&4\omega=3\mathrm{d}_{\mathrm{sp}}-4p^s-6-4l.
\end{align*} 
Now  we have
\[
\begin{aligned}
\omega
&\ge p^s - p^{s{-\tau}} + (\beta - 1)p^{s{-\tau-1}} + 1 \\
\omega
&\ge p^s +p^s - p^{s{-\tau}} + (\beta - 1)p^{s{-\tau-1}} + 1-p^s \\
&\ge p^{\tau+1}+p^{s{-\tau}}(p^{\tau}-1) + (\beta - 1)+1-p^s
\qquad \text{(equality when } \tau = s - 1\text{)} \\
&\ge 2(\beta + 1)p^{\tau} - p{^s}-1
\qquad \text{(equality when } p - 1 = \beta\text{)} \\
&>\mathrm{d}_{\mathrm{sp}}(C)- p^s - 1 \\
& \textrm{that is}, \;4\omega > 4\mathrm{d}_{\mathrm{sp}}-4p^s-4 \\
& 3\mathrm{d}_{\mathrm{sp}}-8p^s-6-4l>4\mathrm{d}_{\mathrm{sp}}-4p^s-4.
\\
\end{aligned}
\]
Therefore  $\mathrm{d}_{\mathrm{sp}}(C)+4l<-2$. This  is
not possible as p is a prime and $l\geq 0$.  So in this case no MDS code exists.
\end{proof}

\begin{theorem}
There do not exist any non-trivial MDS symbol-pair codes among the above \textit{Type 5} of $\alpha$-constacyclic codes of length $4p^s$ over $R_3$.
\end{theorem}
\begin{proof}
For the code $\mathcal{C}$ to reach  the MDS bound, $\mathcal{C}$ must satisfy the symbol-pair singleton bound, that is, $|\mathcal{C}| = p^{3m(4p^s - d_{sp}(\mathcal{C}) + 2)}$. Equating this theoretical bound with the known size of the code, $p^{4m(3p^s - l - U - V)}$, which yields that $p^{4m(3p^s - l - U - V)} = p^{3m(4p^s - d_{sp}(\mathcal{C}) + 2)}$. Thus \begin{align}
    &4(-l - U - V) = 3(-d_{sp}(\mathcal{C}) + 2)\\
    &4(l + U + V) = 3(d_{sp}(\mathcal{C}) - 2)  
    \end{align}
    
\textbf{Case 1:} $U=l$  then from eq (6) we get $4V=3 d_{sp}(C)-8l-6$. Here  we have
 \begin{align*}
\textrm{}~~~~~~~~~~~~ \;V &\geq p^s-p^{s-\tau}+(\beta-1)p^{s-\tau-1}+1 \\
  &= p^s+\big(p^s-p^{s-\tau}+(\beta-1)p^{s-\tau-1}+1-p^s\big) \\
  &\geq 2(\beta+1)p^\tau-p^s-1 \\
   & > d_{sp}(C)-p^s-1\\
  3 d_{sp}(C)-8l-6 &\geq 2d_{sp}(C)-4\\ 
  d_{sp}(C) &\geq 8l+2.
\end{align*}

\textbf{Subcase a:} Assume $V = l$. From the previously established  bounds, we have:

\begin{align*}
V & > d_{sp}(\mathcal{C}) - p^s - 1 \; \textrm{which implies},\\
4V &> 4d_{sp}(\mathcal{C}) - 4p^s - 4.
\end{align*}
Now substituting  $4V = d_{sp}(\mathcal{C}) - 2$ into the left side of this inequality yields:

\begin{align*}
d_{sp}(\mathcal{C}) - 2 &> 4d_{sp}(\mathcal{C}) - 4p^s - 4 \;\textrm{which implies},\\
3d_{sp}(\mathcal{C}) &< 4p^s + 2.
\end{align*}
Given that $p$ is a prime number and $l \ge 0$, this upper bound violates the necessary minimum distance constraints for this code family. Thus there does not exist an MDS code in this scenario.

\textbf{Subcase b:} Assume $V < l$. Utilizing Equation (6), we deduce that $4l > d_{sp}(\mathcal{C}) - 2$ so that  $d_{sp}(\mathcal{C}) < 4l + 2$. So we arrive at the  inequality $4l > 8l$. 

\textbf{Case 2:} Assume $U < l$. By definition of  ideals, $V \le U$; therefore, it strictly follows that $V < l$. Substituting this in Equation (6) yields $4V < 3d_{sp}(\mathcal{C}) - 8l - 6$.  Therefore  from 
$4V \ge 2d_{sp}(\mathcal{C}) - 4$, we  have that, 
\begin{align*}
3d_{sp}(\mathcal{C}) - 8l - 6 &> 2d_{sp}(\mathcal{C}) - 4, \\
\;\textrm{that is,} \;d_{sp}(\mathcal{C}) &> 8l + 2.
\end{align*}
This directly contradicts the fact  $d_{sp}(\mathcal{C}) < 4l + 2$. Indeed  $8l + 2 < 4l + 2$, that is, $4l < 0$.  Consequently, no MDS code can exist under these parameters.
\end{proof}

By applying analogous algebraic bounding techniques, it can be  verified that no non-trivial MDS symbol-pair codes exist among the $\alpha$-constacyclic codes of length $4p^s$ over $R_3$ for ideals of \textit{Type 6}, \textit{Type 7}, and \textit{Type 8}.

\section*{Conclusion}
This paper establishes the symbol-pair distance distributions of all repeated-root $\alpha$-constacyclic codes of length $4p^s$ over the finite commutative chain ring $R_3 = \mathbb{F}_{p^m}[u]/\langle u^3 \rangle$. By evaluating the symbol-pair singleton bound saturation criteria $|\mathcal{C}| = p^{3m(4p^s - d_{sp}(\mathcal{C}) + 2)}$ against the structural code cardinalities, we prove that only the trivial ideal $\mathcal{C} = \langle 1 \rangle$ achieves maximum distance separability (MDS). Conversely, non-trivial MDS symbol-pair codes do not exist within the constituent ideal classes $\mathcal{C}_i$ for $i \ge 2$. 

Future research directions include extending these classification methods to generalized polynomial residue rings $\mathbb{F}_{p^m}[u]/\langle u^k \rangle$ for higher nilpotency indices $k \geq 4$. Additionally, characterizing the distance profiles for alternative repeated-root lengths $2^a p^s$ or $v p^s$ with $\gcd(v,p)=1$ over $R_3$, alongside extensions to the generalized $b$-symbol metric framework, remain open algebraic problems.

\newpage

\begin{center}
  \textbf{Appendix}
\end{center}
Now, we present some examples of constacyclic codes of length $4p^s$ over $\mathbb{F}_{p^m}+u\mathbb{F}_{p^m}+u^2\mathbb{F}_{p^m}$ together with their symbol-pair distances.
\begin{table}[H]
\centering \caption{$\alpha$-constacyclic codes of length $20$ over
$\mathbb{F}_{5} +u \mathbb{F}_{5}+u^2 \mathbb{F}_{5}$}
\begin{longtable}{|l|c|c|}
\hline
\textbf{Ideal (C)} & $\mathbf{d_{sp}}$ & $\mathbf{|C|}$ \\
\hline

\multicolumn{3}{|c|}{Type 1} \\ \hline
$\langle 0 \rangle$ & 0 & 1 \\
$\langle 1 \rangle$ & 2 & $5^3$ \\

\hline \multicolumn{3}{|c|}{Type 2} \\ \hline
$\langle u^2 \rangle$ & 2 & $3^3$ \\
$\langle u^2(x^4-\alpha_{0}) \rangle$ & 4 & $5^2$ \\
$\langle u^2(x^4-\alpha_{0})^2 \rangle$ & 6 & $5$ \\
$\langle u^2(x^4-\alpha_{0})^3 \rangle$ & 8 & $5$ \\
$\langle u^2(x^4-\alpha_{0})^4 \rangle$ & 10 & $5$ \\

\hline \multicolumn{3}{|c|}{Type 3} \\ \hline
$\langle u \rangle$ & 2 & $5^6$ \\
$\langle u(x^4-\alpha_{0}) \rangle$ & 4 & $5^4$ \\
$\langle u(x^4-\alpha_{0})^2 \rangle$ & 6 & $5^2$ \\
$\langle u(x^4-\alpha_{0})^3 \rangle$ & 8 & $5^2$ \\
$\langle u(x^4-\alpha_{0})^2 \rangle$ & 10 & $5^2$ \\
$\langle u(x^4-\alpha_{0})^i + h(x) u^2 \rangle$ [{for i=1,2,3,4}] & 2 & $5^4$ \\
$\langle u(x^4-\alpha_{0})^i + h(x) u^2(x^4-\alpha_{0}) \rangle$ [{for i=1,2,3,4}] & 4 & $5^2$ \\
$\langle u(x^4-\alpha_{0})^i + h(x) u^2(x^4-\alpha_{0})^2 \rangle$ [{for i=2,3,4}] & 6 & $5^2$ \\
$\langle u(x^4-\alpha_{0})^i + h(x) u^2(x^4-\alpha_{0})^3 \rangle$ [{for i=3,4}] & 8 & $5^2$ \\
$\langle u(x^4-\alpha_{0})^4 + h(x)_0 u^2(x^4-\alpha_{0})^4 \rangle $ & 10 & $5^2$ \\
\hline \multicolumn{3}{|c|}{Type 4} \\ \hline
$\langle u(x^4-\alpha_{0})^i, u^2 \rangle$ [{for i=1,2,3,4}] & 2 & $5^5$ \\
$\langle u(x^4-\alpha_{0})^i,h(x) u^2(x^4-\alpha_{0}) \rangle$ [{for i=2,3,4}] & 4 & $5^4$ \\
$\langle u(x^4-\alpha_{0})^i,h(x)u^2(x^4-\alpha_{0})^2 \rangle$ [{for i=3,4}] & 6 & $5^4$ \\
$\langle u(x^4-\alpha_{0})^i+u^2(x^4-\alpha_{0}),h(x) u^2(x^4-\alpha_{0})^2 \rangle$ [{for i=3,4}] & 6 & $5^4$ \\
$\langle u(x^4-\alpha_{0})^4+u^2(x^4-\alpha_{0})^i,h(x) u^2(x^4-\alpha_{0})^3 \rangle$ [{for i=0,1,2}] & 8 & $5^4$ \\

\hline \multicolumn{3}{|c|}{Type 5} \\ \hline
$\langle (x^4-\alpha_{0})^{i}+uh_{1}(x)+u^2h_{2}(x) \rangle$ [{for i=1,2,3,4}] & 4 & $5^6$ \\
$\langle (x^4-\alpha_{0})^{2} +uh_{1}(x)+u^2h_{2}(x)\rangle$ [{for i=2,3,4}] & 4 & $5^3$ \\
$\langle (x^4-\alpha_{0}) ^{i} +u(x^4-\alpha_{0}) ^{t}h_{1}(x)+u^2h_{2}(x) \rangle _{t\epsilon \{0,1\}}$ [{for i=2,3,4}] & 4 & $5^6$ \\
$\langle (x^4-\alpha_{0})^{i} +u(x^4-\alpha_{0}) ^{t}h_{1}(x)+u^2h_{2}(x)\rangle  _{t\epsilon \{0,1,2\}}$ [{for i=3,4}]& 4 & $5^3$ \\
$\langle (x^4-\alpha_{0}) ^{i} +u(x^4-\alpha_{0}) ^{t}h_{1}(x)+u^2(x^4-\alpha_{0}) ^{z}h_{2}(x) \rangle _{t,z \epsilon \{0,1\}}$ [{for i=3,4}] & 6 & $5^6$ \\
$\langle (x^4-\alpha_{0}) ^{i} +u(x^4-\alpha_{0}) ^{t}h_{1}(x)+u^2(x^4-\alpha_{0}) ^{z}h_{2}(x) \rangle _{t \epsilon \{0,1,2\}, z\epsilon \{0,1\}}$ [{for i=3,4}] & 6 & $5^6$ \\
$\langle (x^4-\alpha_{0}) ^{i} +u(x^4-\alpha_{0}) ^{t}h_{1}(x)+u^2(x^4-\alpha_{0}) ^{z}h_{2}(x) \rangle _{t,z \epsilon \{0,1,2\}}$ [{for i=3,4}] & 8 & $5^6$ \\
$\langle (x^4-\alpha_{0})^{4} +u(x^4-\alpha_{0}) ^{t}h_{1}(x)+u^2h_{2}(x)\rangle  _{t\epsilon \{0,1,2,3\}}$ & 4 & $5^3$ \\
$\langle (x^4-\alpha_{0}) ^{4} +u(x^4-\alpha_{0}) ^{t}h_{1}(x)+u^2(x^4-\alpha_{0}) ^{z}h_{2}(x) \rangle _{t \epsilon \{0,1,2,3\}, z\epsilon \{0,1\}}$  & 6 & $5^6$ \\
$\langle (x^4-\alpha_{0}) ^{4} +u(x^4-\alpha_{0}) ^{t}h_{1}(x)+u^2(x^4-\alpha_{0}) ^{z}h_{2}(x) \rangle _{t \epsilon \{0,1,2,3\}, z\epsilon \{0,1,2\}}$  & 8 & $5^6$ \\
$\langle (x^4-\alpha_{0}) ^{4} +u(x^4-\alpha_{0}) ^{t}h_{1}(x)+u^2(x^4-\alpha_{0}) ^{z}h_{2}(x) \rangle _{t,z \epsilon \{0,1,2,3\}}$  & 10 & $5^6$\\
\hline
\end{longtable}

\end{table}

\begin{table}[H]
\centering \caption{Table1(continue) }
\begin{longtable}{|l|c|c|}\hline
\textbf{Ideal (C)} & $\mathbf{d_{sp}}$ & $\mathbf{|C|}$ \\
\hline \multicolumn{3}{|c|}{Type 6} \\ \hline
$\langle (x^4-\alpha_{0})^i, u^2 \rangle$ [{for i=1,2,3,4}] & 2 & $5^7$ \\
$\langle (x^4-\alpha_{0})^i, u^2(x^4-\alpha_{0}) \rangle$ [{for i=2,3,4}] & 4 & $5^5$ \\
$\langle (x^4-\alpha_{0})^i, u^2(x-\alpha_{0})^2 \rangle$ [{for i=3,4}] & 6 & $5^4$ \\
$\langle (x^4-\alpha_{0})^4, u^2(x-\alpha_{0})^3 \rangle$  & 8 & $5^4$ \\
$\langle (x^4-\alpha_{0})^i+a_{0}u, u^2 \rangle$ [{for i=1,2,3,4}] & 2 & $5^7$ \\
$\langle (x^4-\alpha_{0})^i+a_{0}u(x^4-\alpha_{0})^j, u^2 \rangle$ [{for i=2,3,4}, j<i] &2& $5^7$ \\
$\langle (x^4-\alpha_{0})^i+a_{0}u(x^4-\alpha_{0})^j, u^2(x^4-\alpha_{0}) \rangle$ [{for i=2,3,4}, $j<i$] & 4 & $5^7$ \\
$\langle (x^4-\alpha_{0})^i+a_{0}u(x^4-\alpha_{0})^j, u^2(x^4-\alpha_{0})^2 \rangle$ [{for i=3,4}, $j<i$] & 6 & $5^7$ \\
$\langle (x^4-\alpha_{0})^4+a_{0}u(x^4-\alpha_{0})^j, u^2(x^4-\alpha_{0})^3 \rangle$ [{for i=2,3,4}, $j<4$] & 8 & $5^7$ \\
\hline

\multicolumn{3}{|c|}{Type 7} \\ \hline
$\langle (x^4-\alpha_{0})^{i}+uh_{1}(x)+u^2h_{2}(x),u(x^4-\alpha_{0})+u^2h_{3}(x) \rangle$ [{for i=2,3,4}] & 4 & $5^6$ \\
$\langle (x^4-\alpha_{0})^{i}+uh_{1}(x)+u^2h_{2}(x),u(x^4-\alpha_{0})+u^2h_{3}(x) \rangle$ [{for i=2,3,4}] & 4 & $5^6$ \\
$\langle (x^4-\alpha_{0})^{i}+u(x^4-\alpha_{0})^{t}h_{1}(x)+u^2h_{2}(x),u(x^4-\alpha_{0})^{2}$ &&\\$+u^2h_{3}(x) \rangle _{t\epsilon \{0,1\}}$ [{for i=3,4}] & 4 & $5^6$ \\
$\langle (x^4-\alpha_{0})^{i}+u(x^4-\alpha_{0})^{t}h_{1}(x)+u^2(x^4-\alpha_{0})^{z}h_{2}(x),u(x^4-\alpha_{0})^{2}$ &&\\
$+u^2(x^4-\alpha_{0})^{j}h_{3}(x) \rangle _{t,z,j\epsilon \{0,1\}}$ [{for i=3,4}] & 6 & $5^6$ \\
$\langle (x^4-\alpha_{0})^{4}+u(x^4-\alpha_{0})^{t}h_{1}(x)+u^2h_{2}(x),u(x^4-\alpha_{0})^{3}+u^2h_{3}(x) \rangle _{t\epsilon \{0,1,2\}}$  & 4 & $5^6$ \\
$\langle (x^4-\alpha_{0})^{4}+u(x^4-\alpha_{0})^{t}h_{1}(x)+u^2(x^4-\alpha_{0})^{z}h_{2}(x),$ &&\\$u(x^4-\alpha_{0})^{3}+u^2(x^4-\alpha_{0})^{j}h_{3}(x) \rangle _{t\epsilon \{0,1,2\},z,j\epsilon \{0,1\}}$  & 6 & $5^6$ \\
$\langle (x^4-\alpha_{0})^{4}+u(x^4-\alpha_{0})^{t}h_{1}(x)+u^2(x^4-\alpha_{0})^{z}h_{2}(x),$ &&\\$u(x^4-\alpha_{0})^{3}+u^2(x^4-\alpha_{0})^{j}h_{3}(x) \rangle _{t,z,j \epsilon \{0,1,2\}}$  & 6 & $5^6$ \\

\hline \multicolumn{3}{|c|}{Type 8} \\ \hline
$\langle (x^4-\alpha_{0})^2, u(x^4-\alpha_{0}), u^2 \rangle$ & 2 & $5^6$ \\
$\langle (x^4-\alpha_{0})^3, u(x^4-\alpha_{0}), u^2 \rangle$ & 2 & $5^6$ \\
$\langle (x^4-\alpha_{0})^3, u(x^4-\alpha_{0})^2, u^2 \rangle$ & 2 & $5^6$ \\
$\langle (x^4-\alpha_{0})^3, u(x^4-\alpha_{0})^2, u^2(x^4-\alpha_{0}) \rangle$ & 4 & $5^6$ \\
$\langle (x^4-\alpha_{0})^4, u(x^4-\alpha_{0}), u^2 \rangle$ & 2 & $5^6$ \\
$\langle (x^4-\alpha_{0})^4, u(x^4-\alpha_{0})^2, u^2 \rangle$ & 2 & $5^6$ \\
$\langle (x^4-\alpha_{0})^4, u(x^4-\alpha_{0})^2, u^2(x^4-\alpha_{0}) \rangle$ & 4 & $5^6$ \\
$\langle (x^4-\alpha_{0})^4+u(x^4-\alpha_{0}), u(x^4-\alpha_{0})^2, u^2  \rangle$ & 2 & $5^6$ \\
$\langle (x^4-\alpha_{0})^4+u(x^4-\alpha_{0}), u(x^4-\alpha_{0})^2, u^2  \rangle$ & 2 & $5^6$ \\
$\langle (x^4-\alpha_{0})^4+u(x^4-\alpha_{0}), u(x^4-\alpha_{0})^2, u^2(x^4-\alpha_{0})  \rangle$ & 4 & $5^6$ \\
$\langle (x^4-\alpha_{0})^4+u(x^4-\alpha_{0})^i, u(x^4-\alpha_{0})^3, u^2  \rangle$ & 2 & $5^6$ \\
$\langle (x^4-\alpha_{0})^4+u(x^4-\alpha_{0})^i, u(x^4-\alpha_{0})^3, u^2(x^4-\alpha_{0})  \rangle$ & 4 & $5^6$ \\
$\langle (x^4-\alpha_{0})^4+u(x^4-\alpha_{0})^i, u(x^4-\alpha_{0})^3, u^2(x^4-\alpha_{0})^2  \rangle$ & 6 & $5^6$ \\
$\langle (x^4-\alpha_{0})^4+u(x^4-\alpha_{0})^{i}+u^2(x^4-\alpha_{0}), u(x^4-\alpha_{0})^{3}+u^2(x^4-\alpha_{0}), u^2(x^4-\alpha_{0})^2  \rangle$ & 6 & $5^6$ \\
\hline
\end{longtable}
\end{table}

\begin{thebibliography}{99}

\bibitem{Cassuto2011}
Y.~Cassuto, M.~Blaum. "Codes for symbol-pair read channels", \textit{IEEE Transactions on Information Theory}, 57 (2011), pp. 8011-8020.

\bibitem{Cassuto2010}
Y.~Cassuto, M.~Blaum. "Codes for symbol-pair read channels", \textit{Proceedings of IEEE International Symposium on Information Theory}, (2010), pp. 988-992.

\bibitem{Charkani2021}
M.~E.~Charkani, H.~Q.~Dinh, J.~Laaouine, W.~Yamaka. "Symbol-pair distance of repeated-root constacyclic codes of prime power lengths over $\mathbb{F}_{p^m}[u]/\langle u^3 \rangle$", \textit{Mathematics}, 9 (2021), pp. 2554.

\bibitem{Chen2013}
Y.~M.~Chee, L.~Ji, H.~M.~Kiah, C.~Wang, J.~Yin. "Maximum distance separable codes for symbol-pair read channels", \textit{IEEE Transactions on Information Theory}, 59 (2013), pp. 7259-7267.

\bibitem{Chen2016}
B.~Chen, H.~Q.~Dinh, H.~Liu, L.~Wang. "Constacyclic codes of length $2p^s$ over $\mathbb{F}_{p^m} + u\mathbb{F}_{p^m}$", \textit{Finite Fields and Their Applications}, 37 (2016), pp. 108-130.

\bibitem{Dinh2010}
H.~Q.~Dinh. "Constacyclic codes of length $p^s$ over $\mathbb{F}_{p^m} + u\mathbb{F}_{p^m}$", \textit{Journal of Algebra}, 324 (2010), pp. 940-950.

\bibitem{Dinh2017}
H.~Q.~Dinh, B.~T.~Nguyen, A.~K.~Singh, S.~Sriboonchitta. "On the symbol-pair distance of repeated-root constacyclic codes of prime power lengths", \textit{IEEE Transactions on Information Theory}, 64 (2017), pp. 2417-2430.

\bibitem{Dinh2019}
H.~Q.~Dinh, P.~Kumam, P.~Kumar, S.~Satpati, A.~K.~Singh, W.~Yamaka. "MDS symbol-pair repeated-root constacyclic codes of prime power lengths over $\mathbb{F}_{p^m} + u\mathbb{F}_{p^m}$", \textit{IEEE Access}, 7 (2019), pp. 145039-145048.

\bibitem{Dinh2019symbol}
H.~Q.~Dinh, X.~Wang, H.~Lu, S.~Sriboonchitta. "On the symbol-pair distances of repeated-root constacyclic codes of length $2p^s$", \textit{Discrete Mathematics}, 342 (2019), pp. 3062-3078.

\bibitem{Dinh4p}
Dinh, H. Q., Dhompongsa, S.,  Sriboonchitta, S. (2017). On constacyclic codes of length 4ps over Fpm+ uFpm. Discrete Mathematics, 340(4), 832-849.

\bibitem{Dinh2023symbol}
H.~Q.~Dinh, A.~K.~Singh, M.~K.~Thakur. "On symbol-pair distances of repeated-root constacyclic codes of length $2p^s$ over $\mathbb{F}_{p^m} + u\mathbb{F}_{p^m}$ and MDS symbol-pair codes", \textit{Applicable Algebra in Engineering, Communication and Computing}, 34 (2023), pp. 1027-1043.

\bibitem{Dinh2025symbol}
H.~Q.~Dinh, A.~K.~Singh, M.~K.~Thakur. "On symbol-pair distance distributions of repeated-root constacyclic codes of length $4p^s$ and MDS codes", \textit{Advances in Mathematics of Communications}, 19 (2025), pp. 1-21.

\bibitem{Laaouine2021}
J.~Laaouine, M.~E.~Charkani, L.~Wang. "Complete classification of repeated-root $\sigma$-constacyclic codes of prime power length over $\mathbb{F}_{p^m} + u\mathbb{F}_{p^m} + u^2\mathbb{F}_{p^m}$", \textit{Discrete Mathematics}, 344 (2021), pp. 112325.

\bibitem{Laaouine2025}
J.~Laaouine, H.~Q.~Dinh. "On a class of constacyclic codes of length $4p^s$ over $\mathbb{F}_{p^m}[u]/\langle u^3 \rangle$", \textit{Journal of Algebra and Its Applications}, 24 (2025), pp. 2550032.

\bibitem{Ling2004}
S.~Ling, C.~Xing. "Coding Theory: A First Course", \textit{Cambridge University Press}, Cambridge, (2004).

\bibitem{Sriwirach2021}
W.~Sriwirach, C.~Klin-Eam. "Repeated-root constacyclic codes of length $2p^s$ over $\mathbb{F}_{p^m} + u\mathbb{F}_{p^m} + u^2\mathbb{F}_{p^m}$", \textit{Cryptography and Communications}, 13 (2021), pp. 27-52.

\end{thebibliography}
\end{document}